\documentclass[a4paper,11pt,english]{article}
\pdfoutput=1

\usepackage{ifdraft}
\ifdraft{\newcommand{\authnote}[3]{{\color{#3} {\bf  #1:} #2}}}{\newcommand{\authnote}[3]{}}
\newcommand{\anote}[1]{\authnote{ Andras}{#1}{green}}

\newif\ifcount
\ifdraft{\counttrue}

\usepackage[utf8]{inputenc}
\usepackage[T1]{fontenc}
\usepackage[left=1in,top=1in,right=1in,bottom=1in]{geometry} 
\usepackage[english]{babel}
\usepackage{verbatim}

\usepackage{tikz}
\usepackage{amssymb}
\usepackage{mathtools}
\usepackage{bbm}
\usepackage{mleftright}\mleftright
\usepackage{multirow}
\usepackage{amsthm} 
\theoremstyle{plain}
\usepackage{thmtools} 
\usepackage{thm-restate}

\usepackage{algorithm}
\usepackage{algcompatible}
\usepackage{float}
\usepackage{enumitem}
\makeatletter
\def\namedlabel#1#2{\begingroup
	#2%
	\def\@currentlabel{#2}%
	\phantomsection\label{#1}\endgroup
}

\usepackage{caption}
\usepackage{subcaption}

\usepackage{titling}
\usepackage[final=true,colorlinks = true,allcolors = {blue},]{hyperref}

\newcommand{\eps}{\varepsilon}

\newcommand{\ipc}[2]{\langle #1 , #2 \rangle}
\newcommand{\ketbra}[2]{|#1\rangle\! \langle #2|}
\newcommand{\braketbra}[3]{\langle #1|#2| #3 \rangle}
\newcommand{\nrm}[1]{\left\lVert #1 \right\rVert}
\newcommand{\bigO}[1]{\mathcal{O}\left( #1 \right)}
\newcommand{\bigOt}[1]{\widetilde{\mathcal{O}}\left( #1 \right)}

\newcommand{\vertiii}[1]{{\left\vert\kern-0.25ex\left\vert\kern-0.25ex\left\vert #1 
		\right\vert\kern-0.25ex\right\vert\kern-0.25ex\right\vert}}

\newcommand{\C}{\mathbb{C}}

\newcommand{\R}{\mathbb{R}}

\newcommand{\pvp}{\vec{p}{\kern 0.45mm}'}
\let\oldnabla\nabla
\renewcommand{\nabla}{\oldnabla\!}

\DeclarePairedDelimiter\bra{\langle}{\rvert}
\DeclarePairedDelimiter\ket{\lvert}{\rangle}
\DeclarePairedDelimiterX\braket[2]{\langle}{\rangle}{#1 \delimsize\vert #2}
\newcommand{\underflow}[2]{\underset{\kern-60mm \overbrace{#1} \kern-60mm}{#2}}

\def\polylog{\mathrm{polylog}}

\def\Tr{\mathrm{Tr}}

\providecommand{\tr}[1]{\Tr\left[#1\right]}

\providecommand{\rank}[1]{\mathrm{rank}\left(#1\right)}
\providecommand{\poly}[1]{\mathrm{poly}\left(#1\right)}

\long\def\ignore#1{}

\newtheorem{theorem}{Theorem}
\newtheorem{lemma}[theorem]{Lemma}

\newtheorem{definition}[theorem]{Definition}

\newtheorem*{claim*}{Claim}

\usepackage{tikz}

\usetikzlibrary{backgrounds,fit,decorations.pathreplacing,calc}

\title{	
	Quantum-inspired low-rank stochastic regression with logarithmic dependence on the dimension
}
\author{
	András Gilyén\thanks{QuSoft, CWI and University of Amsterdam, the Netherlands. Supported by ERC Consolidator Grant QPROGRESS and partially supported by QuantERA project QuantAlgo 680-91-034.}\textsuperscript{\kern1.4mm,\textasteriskcentered}
	\and
    Seth Lloyd\thanks{Massachusetts Institute of Technology, Departments of Mechanical Engineering and Physics, Xanadu}
	\and
	Ewin Tang\thanks{University of Washington}    
}
\thanksmarkseries{arabic}
\date{\today\vspace{-5mm}}

\newcommand{\astfootnote}[1]{
	\let\oldthefootnote=\thefootnote
	\setcounter{footnote}{0}
	\renewcommand{\thefootnote}{\fnsymbol{footnote}}
	\footnotetext{#1}
	\let\thefootnote=\oldthefootnote
}

\begin{document}
	
\maketitle \astfootnote{\kern-1.4mm\textsuperscript{\kern1mm\textasteriskcentered}Corresponding author: \texttt{gilyen@cwi.nl}}

\begin{abstract}
	We construct an efficient classical analogue of the quantum matrix inversion algorithm \cite{harrow2009QLinSysSolver} for low-rank matrices.
	Inspired by recent work of Tang \cite{tang2018QuantumInspiredRecommSys}, assuming length-square sampling access to input data, we implement the pseudoinverse of a low-rank matrix and sample from the solution to the problem $Ax=b$ using fast sampling techniques.
	We implement the pseudo-inverse by finding an approximate singular value decomposition of $A$ via subsampling, then inverting the singular values.
	In principle, the approach can also be used to apply any desired ``smooth'' function to the singular values.
	Since many quantum algorithms can be expressed as a singular value transformation problem~\cite{gilyen2018QSingValTransf}, our result suggests that more low-rank quantum algorithms can be effectively ``dequantised'' into classical length-square sampling algorithms.
\end{abstract}
	
\section{Introduction}

	Quantum computing provides potential exponential speed-ups over classical computers for a variety of linear algebraic tasks, including an operational version of matrix
	inversion~\cite{harrow2009QLinSysSolver}.
	Recently, inspired by the quantum algorithm for recommendation systems~\cite{kerenidis2016QRecSys}, Tang showed how to generalize the well-known FKV algorithm~\cite{frieze2004FastMonteCarloLowRankApx} to sample from the singular vectors of low-rank matrices~\cite{tang2018QuantumInspiredRecommSys} and to implement principal component component analysis~\cite{tang2018QInspiredClassAlgPCA}.
	Intriguingly, Tang's work suggests that many of the quantum algorithms for low-rank matrix manipulation \cite{rebentrost2016QuantumSVDNonsparseLowRank} can be extended to provide fast {\it classical} algorithms under suitable sampling assumptions, achieving logarithmic dependence on the dimension.
	In this work, we show that such exponential speed-ups are indeed possible in the case of low-rank matrix inversion.
	Our treatment is self-contained and improves some aspects of previous approaches~\cite{tang2018QuantumInspiredRecommSys}, leading to smaller exponents in our runtime bounds.

	Suppose we want to solve $Ax=b$, where we are given $A\in\R^{m \times n}$ and $b\in\R^{m}$, and wish to recover $x\in\R^{n}$. The equation might not have a solution, but we can always find an $x$ minimizing $\nrm{Ax-b}$. Namely, $x = A^+b$ works, where $A^+$ is the pseudoinverse of $A$. If 
	\begin{equation*}
	A=\sum_{\ell=1}^{k}\sigma_\ell\, u^{(\ell)}{v^{(\ell)}}^T
	\end{equation*}
	is a singular value decomposition of $A$, such that $\sigma_\ell>0$, then the pseudoinverse is simply $A^+=\sum_{\ell=1}^{k}\, v^{(\ell)}{u^{(\ell)}}^T/\sigma_\ell$, and $x=\sum_{\ell=1}^k v^{(\ell)} \ipc{u^{(\ell)}}{b}/\sigma_\ell$. The problem of finding the pseudoinverse of a low-rank matrix occurs widely, such as in problems of data fitting, stochastic regression, and quadratic optimization with linear constraints. 	 
	
	Applications of the classical stochastic regression algorithm presented here include a wide variety of data analysis problems.   Consider, for example, the problem of finding the optimal	investment portfolio amongst $n$ stocks.	Let $r_i$ be the  $m$ vectors of historical returns on the stocks (e.g., the returns on the $i$'th day or the $i$'th tick of the	stock market), and let $\vec r = (1/m) \sum_i r_i$, so that $A$ is the	matrix with rows $r_i^T$.    The covariance matrix is $C= \sum_i r_i r_i^T/m = A^TA$.	Typically, the covariance matrix $C$ is (approximately) low-rank, with each singular value corresponding to an underlying trend of correlated motion of investment returns.	Classical portfolio management operates by finding the vector of investments $w$ that maximizes the	expected return $w^T r$ subject to a constraint on the variance $w^T C w$.	As shown in \cite{rebentrost2018QuantumFinance}, finding and sampling from the optimal portfolio and mapping	out the optimal risk-return curve is a low-rank matrix inversion problem which can be solved on a quantum computer in time $\bigO{\polylog(mn)}$,	exponentially faster than conventional classical portfolio optimization methods.    The results presented here show that our quantum-inspired classical algorithm can similarly allow one to map out the risk-return curve and sample from the optimal portfolio using $\bigO{\polylog(mn)}$ {\it classical} time.

\section{The algorithm}
	We use the following notation. For $v\in\C^d$ we denote the Euclidean norm by $\nrm{b}$. For a matrix $A\in\C^{m\times n}$ we denote by $\nrm{A}$ the operator norm, and by $\nrm{A}_F$ the Frobenius norm. We use notation $A_{i.}$ for the $i$-th row, $A_{.j}$ for the $j$-th column, and $A^\dagger$ for the adjoint of $A$. We use ``bra-ket'' notation: for $v\in \C^d$ we denote the corresponding column vector by $\ket{v}\in\C^{d\times 1}$, and we denote by $\bra{v}\in \C^{1\times d}$ its adjoint. Accordingly we denote the inner product $\ipc{v}{w}$ by $\braket{v}{w}$. We call the probability distribution $|v_i|^2/\nrm{v}^2$ over $i\in[d]$ the \emph{length-square distribution} of $v$.

	In this paper for simplicity we treat the case when the matrix $A$ has rank $k\ll m,n$, and does not have too small singular values.\footnote{This assumption can be relaxed by inverting $A$ only on the ``well-conditioned'' subspace, and dealing with small singular values similarly to the earlier works~\cite{frieze2004FastMonteCarloLowRankApx,tang2018QuantumInspiredRecommSys}, but we leave such details for future work.} For normalization purposes we assume that $\nrm{A}\leq 1$ and $\nrm{A^+}\leq \kappa$. 
	Our program is the following: we first show how to describe an approximate singular value decomposition $\sum_{\ell=1}^{k}\tilde{\sigma}_\ell\,\ketbra{\tilde{u}^{(\ell)}}{\tilde{v}^{(\ell)}}$ using a succinct representation of the vectors. Then we show how to estimate the values $\braket{\tilde{u}^{(\ell)}}{b}/\tilde{\sigma}_\ell$ via sampling, and how to sample from the corresponding linear combination of the vectors $\tilde{v}^{(\ell)}$. The overall algorithm allows us to sample or query elements of an approximate solution $\tilde{x}\approx A^+ b$ to the equation $Ax=b$ in time poly-logarithmic in the size of the matrix.

	The idea of the approximate singular value decomposition of $A$ using length-square sampling comes from~\cite{frieze2004FastMonteCarloLowRankApx}. Consider the following: first use length-square sampling to sample some rows $R$. If we sample enough rows, then $\nrm{A^T A -R^T R}$ is small as shown by Theorem~\ref{thm:fastSample}.
	Then the singular values and right singular vectors of $R$ are very close to the singular values and right singular vectors of $A$.   Using the approximate right singular vectors we can recover approximate left singular vectors as well, e.g., by applying the matrix $A$ to the right singular vector, as shown by Lemma~\ref{lem:SVecConv}. This is promising, but since $R$ still can have very high-dimensional rows, computing the singular value decomposition of $R$ can still be prohibitive. However, we can apply the trick once more! We can sample columns of $R$ resulting in the matrix $C$. Again the singular values and left singular vectors of $C$ are very close to the singular values and left singular vectors of $R$ provided that $\nrm{RR^T -CC^T}$ is small.

	\begin{theorem}[Correctness]\label{thm:main}
		If $A$ has rank at most\footnote{While running the algorithm we can actually detect if $A$ has lower rank and adapt the algorithm accordingly.} $k$, $\nrm{A}\leq 1$, $\nrm{A^+}\leq\kappa$, and the projection of $b$ to the column space of $A$ has norm $\Omega(\nrm{b})$, then Algorithm~\ref{alg:LowRankHHL} solves $Ax=b$ up to $\eps$-multiplicative accuracy, such that $\nrm{\tilde{x}-A^+b}\leq\eps\nrm{A^+b}$ with probability at least $1-\eta$.
	\end{theorem}

	In order to execute Algorithm~\ref{alg:LowRankHHL} we use length-square sampling techniques, which have found great applications in randomized linear algebra~\cite{kannan2017RandAlgNumLinAlg}, as well as recent quantum-inspired classical algorithms for recommendation systems~\cite{tang2018QuantumInspiredRecommSys}. For simplicity we assign unit cost to arithmetic operations such as addition or multiplication of real or complex numbers, assuming that all numbers are represented with a small number of bits. 
	If $A\in\C^{m\times n}$ is stored in an efficient tree-like data-structure as in~\cite{frieze2004FastMonteCarloLowRankApx,kerenidis2016QRecSys,tang2018QuantumInspiredRecommSys}, then we can implement the sampling and query operations of $A$ required by Algorithm~\ref{alg:LowRankHHL} in complexity $\bigO{\log(mn)}$ assuming that data-structure queries have unit cost. Under these assumption we get the following bound:
	
	\begin{theorem}[Complexity]\label{thm:cplx}
		If we have $\bigOt{1}$-time\footnote{In this paper by $\bigOt{T}$ we hide poly-logarithmic factors in $T$, the dimensions $m,n$ and the failure probability~$\eta$.} query access to $b$ and $\bigOt{1}$-time length-square access to $A$, then under the conditions of Theorem~\ref{thm:main} we can execute Algorithm~\ref{alg:LowRankHHL} in complexity $\bigOt{\kappa^{16} k^6\nrm{A}_F^6/\eps^6}$, outputting an implicit description of $\tilde{x}$ suitable for query and sample access.
	\end{theorem}
	
	Note that our complexity bound has smaller exponents than e.g.~\cite{tang2018QuantumInspiredRecommSys}. This partly comes from the fact that we only consider low-rank matrices, but we also get improvements by adapting and reanalysing the FKV algorithm~\cite{frieze2004FastMonteCarloLowRankApx}. We only work out the constant factors for the number of rows and columns to be sampled, because these parameters dominate the complexity.
	
	For comparison with the quantum analogue, note that under the assumption that the data structure for $A$ is stored in quantum memory, an $\eps$-approximate quantum state $\ket{\tilde{x}}/\nrm{\tilde{x}}$ can be prepared in complexity $\bigOt{\kappa\nrm{A}_F\polylog(1/\eps)}$, as shown in~\cite{chakraborty2018BlockMatrixPowers,gilyen2018QSingValTransf}. This directly enables length-square sampling, and its entries can be estimated with $\poly{\kappa/\eps}$ overheads.
	
	\begin{algorithm}[H]
		\caption{Low-rank stochastic regression via length-square sampling}\label{alg:LowRankHHL}
		\begin{algorithmic}[1]
			\STATEx {\bf Input:} A vector $b\in\C^m$ and a matrix $A\in\C^{m\times n}$ s.t. $\nrm{A}\leq 1$, $\rank{A}=k$ and $\nrm{A^+}\leq\kappa$.
			\STATEx {\bf Goal 1:} Query elements of a vector $\tilde{x}$ such that $\nrm{\tilde{x}-x}\leq \eps\nrm{x}$ for $x=A^+b$.\label{line:goal1}	
			\STATEx {\bf Goal 2:} Sample from a distribution $2\eps$-close in total-variation distance to $\frac{|x_j|^2}{\nrm{x}^2}$.\label{line:goal2}								
			\STATE {\bf Init: } Set $r=2^{10}\ln\left(\frac{8n}{\eta}\right)\frac{\kappa^4 k^2\nrm{A}_F^2}{\eps^2}$ and $c=2^{6}\cdot3^4\ln\left(\frac{8r}{\eta}\right)\frac{\kappa^8 k^2\nrm{A}_F^2}{\eps^2}$.\label{line:init}		
			\STATE{\bf Sample rows:} Sample $r$ row indices $i_1,i_2,\ldots, i_r$ according to the row norm squares $\frac{\nrm{A_{i.}}^2}{\nrm{A}_F^2}$. \label{line:rowSample}
			\vspace{-4mm}
			\STATEx Define $R$ to be the matrix whose $s$-th row is $\frac{\nrm{A}_F}{\sqrt{r}}\frac{A_{i_s.}}{\nrm{A_{i_s.}}}$.
			\STATE{\bf Sample columns:} Sample $s\in [r]$ uniformly, then sample a column index $j$ distributed as $\frac{|R_{s j}|^2}{\nrm{R_{s .}}^2}=\frac{|A_{i_s j}|^2}{\nrm{A_{i_s .}}^2}$. Sample a total number of $c$ column indices $j_1,j_2,\ldots, j_c$ this way. 
			Define the matrix $C$ whose $t$-th column is $\frac{\nrm{R}_F}{\sqrt{c}}\frac{R_{.j_t}}{\nrm{R_{.j_t}}}=\frac{\nrm{A}_F}{\sqrt{c}}\frac{R_{.j_t}}{\nrm{R_{.j_t}}}$. \label{line:columnSample}			
			\STATE{\bf SVD:} Query all elements of $A$ corresponding to elements of $C$.	Compute the left singular vectors and singular values of $C$ and denote the left singular vectors by $w^{(1)},\ldots, w^{(k)}$ and the corresponding singular values by $\tilde{\sigma}_1,\ldots, \tilde{\sigma}_k$.
			\STATE{\bf Approximate right singular vectors of $A$:} Implicitly define $\tilde{v}^{(\ell)}:=\left(\sum_{s=1}^r R_{i_s.}^\dagger \frac{w^{(\ell)}_s}{\tilde{\sigma}_\ell} \right)$.
			\STATE{\bf Matrix elements:} For each $\ell\in[k]$ compute $\tilde{\lambda}_\ell$ such that $\left|\tilde{\lambda}_\ell-\braketbra{\tilde{v}^{(\ell)}}{A^\dagger}{b}\right|=\bigO{\frac{\eps\tilde{\sigma}_\ell^2\nrm{b}}{\sqrt{k}}}$.
			\STATE{\bf Output:} Row indices $i_1,i_2,\ldots, i_r$ and $w:=\frac{\tilde{\lambda}_\ell}{\tilde{\sigma}_\ell^3} w^{(\ell)}\in\C^r$ such that $\nrm{w}=\bigO{\kappa^2\sqrt{k}\nrm{b}}$.
			\STATEx{\bf Queries to $\tilde{x}$:} Define $\tilde{x}:=R^\dagger w$, where $R$ is implicitly defined by the row indices. A query to $\tilde{x}_j$ can be computed by querying $R_{s,j}$ for all $s\in[r]$ and taking their linear combination.	
			\STATEx{\bf Sampling from $|\tilde{x}_j|^2/\nrm{\tilde{x}}^2$:} Rejection sample (Lemma~\ref{lem:linCombSamp}), using $T$ samples from the distribution $\frac{|R_{s,j}|^2}{\nrm{R_{s,.}}^2}$ for some $s\in [r]$, and querying $rT$ entries of $R$, s.t. $\mathbb{E}[T]\leq \nrm{w}^2\nrm{A}_F^2/\nrm{\tilde{x}}^2$.
		\end{algorithmic}
	\end{algorithm}
	In the above algorithm, we first convert left singular vectors of $C$ ($w^{(\ell)}$) to approximate right singular vectors of $R$ ($\tilde{v}^{(\ell)}$), which also approximate right singular vectors of $A$. Then we ``convert'' these to left singular vectors of $A$ in the form ($\bra{\tilde{v}^{(\ell)}}A^\dagger/\tilde{\sigma}_{\ell}$).
	To clarify the formula for $\tilde{x}$, notice the following sequence of approximations:
	$$\tilde{x} \approx \sum_{\ell=1}^k \frac{1}{\tilde{\sigma}_{\ell}^4}\ketbra{R^\dagger w^{(\ell)}}{R^\dagger w^{(\ell)}}A^\dagger b \approx \sum_{\ell=1}^k \frac{1}{\tilde{\sigma}_\ell^2}\ketbra{\tilde{v}^{(\ell)}}{\tilde{v}^{(\ell)}} A^\dagger b \approx (R^\dagger R)^+A^\dagger b \approx (A^\dagger A)^+A^\dagger b = A^+b$$
	The conversion step from right to left singular vectors of $A$ magnifies previous inaccuracies.
	For this reason, it is beneficial to sample a higher number of columns than rows, unlike in earlier works~\cite{frieze2004FastMonteCarloLowRankApx,tang2018QuantumInspiredRecommSys}.

\section{Correctness of Algorithm~\ref{alg:LowRankHHL}}
	
	The goal is to output a description of an approximate solution $\tilde{x}$. If $\nrm{\tilde{x}-x}\leq \eps\nrm{x}$, then the length-square distributions of $x$ and $\tilde{x}$ are $2\eps$-close as shown by Lemma~\ref{lem:TVEuclid}. 
	Thus for our purposes it suffices to find approximate right singular vectors $\tilde{v}^{(\ell)}$ and approximate singular values $\tilde{\sigma}_\ell$ such that 
	\begin{equation}\label{eq:goodEnoughSing}
	\nrm{\sum_{\ell=1}^k\frac{\ketbra{\tilde{v}^{(\ell)}}{\tilde{v}^{(\ell)}}}{\tilde{\sigma}_\ell^2}A^\dagger A - \Pi_{\mathrm{rows}(A)}}\leq \frac{\eps}{2}.
	\end{equation}
	Let us define 
	\begin{equation}\label{eq:goodEnoughSingAgain}
	\ket{x'}=\sum_{\ell=1}^k\frac{\ketbra{\tilde{v}^{(\ell)}}{\tilde{v}^{(\ell)}}}{\tilde{\sigma}_\ell^2}A^\dagger \ket{b}
	=\sum_{\ell=1}^k\frac{\braketbra{\tilde{v}^{(\ell)}}{A^\dagger}{b}}{\tilde{\sigma}_\ell^2} \ket{\tilde{v}^{(\ell)}}
	=\sum_{\ell=1}^k\frac{\lambda_\ell}{\tilde{\sigma}_\ell^2} \ket{\tilde{v}^{(\ell)}},
	\end{equation}	
	then due to Equation~\eqref{eq:goodEnoughSing} we have that
	\begin{equation*}
	\nrm{x'-x}=
	\nrm{\sum_{\ell=1}^k\frac{\ketbra{\tilde{v}^{(\ell)}}{\tilde{v}^{(\ell)}}}{\tilde{\sigma}_\ell^2}A^\dagger \ket{b} - \ket{x}}
	=	\nrm{\!\left(\sum_{\ell=1}^k\frac{\ketbra{\tilde{v}^{(\ell)}}{\tilde{v}^{(\ell)}}}{\tilde{\sigma}_\ell^2}A^\dagger A - \Pi_{\mathrm{rows}(A)}\!\right)\ket{x}}
	\leq \frac{\eps}{2}\nrm{x}.
	\end{equation*}
	Remember that we assumed that the projection of $b$ to the column space of $A$ has norm $\Omega(\nrm{b})$. Since $\nrm{A}\leq 1$ we also have that $\nrm{x}=\Omega(\nrm{b})$. Therefore it suffices to find $\tilde{x}$ such that $\nrm{\tilde{x}-x'}=\bigO{\eps\nrm{b}}$ in order to ensure $\nrm{\tilde{x}-x'}\leq \frac{\eps}{2}\nrm{x}$.	
  	If we compute approximate values $\tilde{\lambda}_\ell$ such that $\left|\lambda_\ell-\tilde{\lambda}_\ell\right|=\bigO{\frac{\eps\tilde{\sigma}_\ell^2\nrm{b}}{\sqrt{k}}}$ then the magnitude of perturbation $\nrm{\tilde{x}-x'}$ can be bounded by $\bigO{\eps\nrm{b}}$, and $\nrm{w}=\bigO{\kappa^2 k \nrm{b}}$ as we show in Section~\ref{subsec:precReq}.

\subsection{Finding approximate singular values and right singular vectors}\label{subsec:corr}
	First we invoke some improved bounds on length-square sampling from~\cite[Theorem 4.4]{kannan2017RandAlgNumLinAlg}.\footnote{In~\cite{kannan2017RandAlgNumLinAlg} the theorem is stated for real matrices, but the proof works for complex matrices as well.}
	
	Length-square row sampling of a matrix $A\in\C^{m \times n}$ is as follows: pick a row index $i\in[m]$ with probability $p_i=\frac{\nrm{A_{i.}}^2}{\nrm{A}^2_F}$, and upon picking index $i$ set the random output $Y=\frac{A_{i.}}{\sqrt{p_i}}$.
    Notice that in Algorithm~\ref{alg:LowRankHHL} both $R$ and $C$ can be characterized as length-square (row) sampled matrices (the latter holding because every row of $R$ has the same norm).
	\begin{theorem}\label{thm:fastSample}
		Let $A\in\C^{m \times n}$ be a matrix and let $R\in\C^{s \times n}$ be the sample matrix obtained by length-squared sampling and scaling to have $\mathbb{E}[R^\dagger R]= A^\dagger A$. ($R$ consists of rows $Y_1 , Y_2 , \ldots , Y_s$, which are i.i.d. copies of $Y/\sqrt{s}$, as defined above.) Then, for all $\eps \in [0, \nrm{A}/\nrm{A}_F]$,\footnote{If $\eps \geq \nrm{A}/\nrm{A}_F$, then the zero matrix is a good enough approximation to $AA^\dagger$.}  we have 
		\begin{equation*}
			\mathbb{P}\left[\nrm{R^\dagger R-A^\dagger A}\geq \eps \nrm{A}\nrm{A}_F\right]\leq 2n e^{-\frac{\eps^2 s}{4}}.
		\end{equation*}
		Hence, for $s \geq \frac{4\ln(2n/\eta)}{\eps^2}$, with probability at least $(1-\eta)$ we have
		\begin{equation*}
			\nrm{R^\dagger R-A^\dagger A}\leq \eps \nrm{A}\nrm{A}_F.
		\end{equation*}	
	\end{theorem}
		
	In the following lemma $\nrm{M}$ denotes the operator norm, but the proof would also work for the Frobenius norm. Note that the following lemmas are independent of the dimensions of the matrices, which is the reason why we do not specify the dimensions. We use $\delta_{ij}$ to denote the Kronecker delta function, which is defined to be $1$ if $i = j$ and $0$ otherwise.
	\begin{lemma}[Converting approximate left and right singular vectors]\label{lem:SVecConv}
		Suppose that $w^{(\ell)}$ is a system of orthonormal vectors spanning the column space of $C$ such that 
		\begin{equation*}
		\sum_{\ell=1}^{k}\ketbra{w^{(\ell)}}{w^{(\ell)}}= \Pi_{\mathrm{cols}(C)}
		\end{equation*}
		and
		\begin{equation*}
		\bra{w^{(i)}} C C^\dagger \ket{w^{(j)}}=\delta_{ij}\tilde{\sigma}_i^2.
		\end{equation*}		
		Suppose that $\rank{R}=\rank{C}=k$ and $\nrm{RR^\dagger-CC^\dagger}\leq \gamma$.
		Let $\tilde{v}^{(\ell)}:=\frac{R^\dagger w^{(\ell)}}{\tilde{\sigma}_\ell}$, then 
		\begin{equation*}
		|\braket{\tilde{v}^{(i)}}{\tilde{v}^{(j)}}-\delta_{ij}|\leq \frac{\gamma}{\tilde{\sigma}_i\tilde{\sigma}_j},
		\end{equation*}
		and
		\begin{equation*}
		\left|\bra{\tilde{v}^{(i)}}R^\dagger R\ket{\tilde{v}^{(j)}}-\delta_{ij}\tilde{\sigma}_i^2\right|\leq 
		\frac{\gamma\left(2\nrm{R}^2+\gamma\right)}{\tilde{\sigma}_i\tilde{\sigma}_j}.
		\end{equation*}			
	\end{lemma}
	\begin{proof}
		Let $V$ be the matrix whose $\ell$-th column is the vector $\tilde{v}^{(\ell)}$ and let us define the Gram matrix $G=V^\dagger V$. We have that 
		\begin{align*}
		G_{ij}=|\ipc{\tilde{v}^{(i)}}{\tilde{v}^{(j)}}-\delta_{ij}|
		&=\left|\frac{\bra{w^{(i)}} R R^\dagger \ket{w^{(j)}}}{\tilde{\sigma}_i\tilde{\sigma}_j}-\delta_{ij}\right|\nonumber\\
		&\leq \left|\frac{\bra{w^{(i)}} C C^\dagger \ket{w^{(j)}}}{\tilde{\sigma}_i\tilde{\sigma}_j}-\delta_{ij}\right|+ \frac{\gamma}{\tilde{\sigma}_i\tilde{\sigma}_j}\nonumber\\
		&= \frac{\gamma}{\tilde{\sigma}_i\tilde{\sigma}_j}.
		\end{align*}			
		Now observe that
		$$ \nrm{RR^\dagger RR^\dagger-CC^\dagger CC^\dagger}
		\leq \nrm{RR^\dagger(RR^\dagger-CC^\dagger)} + \nrm{(RR^\dagger-CC^\dagger) CC^\dagger}\leq \gamma\left(2\nrm{R}^2+\gamma\right).
		$$
		Let $i,j\in[k]$, then 
		\begin{align*}
		\bra{w^{(i)}}CC^\dagger CC^\dagger \ket{w^{(j)}}
		&= \bra{w^{(i)}}CC^\dagger\left(\sum_{\ell=1}^{k}\ketbra{w^{(\ell)}}{w^{(\ell)}}\right) CC^\dagger \ket{w^{(j)}}=\delta_{ij}\sigma_i^4.
		\end{align*}
		Finally we get that
		\begin{align*}
			\left|\bra{\tilde{v}^{(i)}}R^\dagger R\ket{\tilde{v}^{(j)}}-\delta_{ij}\sigma_i^2\right|&=\left|\frac{\bra{w^{(i)}}RR^\dagger RR^\dagger \ket{w^{(j)}}}{\tilde{\sigma}_i\tilde{\sigma}_j}-\delta_{ij}\sigma_i^2\right|\\
			&\leq\left|\frac{\bra{w^{(i)}}CC^\dagger CC^\dagger \ket{w^{(j)}}}{\tilde{\sigma}_i\tilde{\sigma}_j}-\delta_{ij}\sigma_i^2\right| + 
			\frac{\gamma\left(2\nrm{R}^2+\gamma\right)}{\tilde{\sigma}_i\tilde{\sigma}_j}\\
			&=\frac{\gamma\left(2\nrm{R}^2+\gamma\right)}{\tilde{\sigma}_i\tilde{\sigma}_j}.\qedhere
		\end{align*}			
	\end{proof}

	\begin{lemma}\label{lem:SpecBoundApx}
		Let $B$ be a matrix of rank at most $k$, and suppose that $V$ has $k$ columns that span the row and column spaces of $B$. Then 
		$$
		\nrm{B}\leq \nrm{(V^\dagger V)^{-1}}\nrm{V^\dagger B V}.
		$$
	\end{lemma}
	\begin{proof}
		Let $G:=V^\dagger V$ be the Gram matrix of $V$ and let $\tilde{V}:=VG^{-\frac{1}{2}}$. It is easy to see that $\tilde{V}$ is an isometry and its columns still span the the row and column spaces of $B$. Since $\tilde{V}$ is an isometry we get that 
		\begin{equation*}
			\nrm{B}=\nrm{\tilde{V}^\dagger B \tilde{V}}=\nrm{G^{-\frac{1}{2}}V^\dagger B VG^{-\frac{1}{2}}}\leq \nrm{G^{-1}}\nrm{V^\dagger B V}=\nrm{(V^\dagger V)^{-1}}\nrm{V^\dagger B V}.\qedhere
		\end{equation*}
	\end{proof}

	\begin{lemma}[Approximate left and right singular vectors]\label{lem:ApxProj}
		Suppose that $\tilde{v}^{(i)}$ is a system of approximately orthonormal vectors spanning the row space of $A$ such that 
		\begin{equation}\label{eq:apxVGram}
			|\braket{\tilde{v}^{(i)}}{\tilde{v}^{(j)}}-\delta_{ij}|\leq \alpha\leq \frac{1}{4k},
		\end{equation}
		and
		\begin{equation*}
			\left|\bra{\tilde{v}^{(i)}}R^\dagger R\ket{\tilde{v}^{(j)}}-\delta_{ij}\tilde{\sigma}_i^2\right|\leq 
			\beta,
		\end{equation*}	
		where $\tilde{\sigma}^2_i \geq \frac{4}{5\kappa^2}$.
		Suppose that $\rank{A}=\rank{R}=k$ and $\nrm{A^\dagger A-R^\dagger R}\leq \theta$, then 
		\begin{equation}\label{eq:ultimateBound}
			\nrm{\Pi_{\mathrm{rows}(A)}-\sum_{\ell=1}^{k}\frac{\ketbra{\tilde{v}^{(\ell)}}{\tilde{v}^{(\ell)}}}{\tilde{\sigma}_\ell^2}A^\dagger A}\leq \frac{8k}{3}(\beta\kappa^2+\theta\kappa^2+ \alpha).
		\end{equation}
	\end{lemma}
	\begin{proof}
		Let $B:=\sum_{\ell=1}^{k}\frac{\ketbra{\tilde{v}^{(\ell)}}{\tilde{v}^{(\ell)}}}{\tilde{\sigma}_\ell^2}A^\dagger A -\Pi_{\mathrm{rows}(A)}$, we will apply Lemma~\ref{lem:SpecBoundApx}. For this observe
		\begin{align*}
			\left|\bra{\tilde{v}^{(i)}}B\ket{\tilde{v}^{(j)}}\right|
			&=\left|\sum_{\ell=1}^{k}\frac{\braket{\tilde{v}^{(i)}}{\tilde{v}^{(\ell)}}\braketbra{\tilde{v}^{(\ell)}}{A^\dagger A}{\tilde{v}^{(j)}}}{\tilde{\sigma}_\ell^2}- \braket{\tilde{v}^{(i)}}{\tilde{v}^{(j)}}\right|\\
			&\leq\left|\sum_{\ell=1}^{k}\frac{\braket{\tilde{v}^{(i)}}{\tilde{v}^{(\ell)}}\braketbra{\tilde{v}^{(\ell)}}{R^\dagger R}{\tilde{v}^{(j)}}}{\tilde{\sigma}_\ell^2}-\delta_{ij}\right|
			+\sum_{\ell=1}^{k}\frac{|\braket{\tilde{v}^{(i)}}{\tilde{v}^{(\ell)}}|\theta\nrm{\tilde{v}^{(\ell)}}\nrm{\tilde{v}^{(j)}}}{\tilde{\sigma}_\ell^2}+ \alpha\\
			&\leq\left|\sum_{\ell=1}^{k}\frac{\braket{\tilde{v}^{(i)}}{\tilde{v}^{(\ell)}}\braketbra{\tilde{v}^{(\ell)}}{R^\dagger R}{\tilde{v}^{(j)}}}{\tilde{\sigma}_\ell^2}-\delta_{ij}\right|
			+2\theta\kappa^2+ \alpha\\	
			&\leq\left|\sum_{\ell\neq j}^{k}\frac{\braket{\tilde{v}^{(i)}}{\tilde{v}^{(\ell)}}\braketbra{\tilde{v}^{(\ell)}}{R^\dagger R}{\tilde{v}^{(j)}}}{\tilde{\sigma}_\ell^2}\right|
			+\left|\braket{\tilde{v}^{(i)}}{\tilde{v}^{(j)}}\frac{\braketbra{\tilde{v}^{(j)}}{R^\dagger R}{\tilde{v}^{(j)}}}{\tilde{\sigma}_j^2}-\delta_{ij}\right|
			+2\theta\kappa^2+ \alpha\\	
			&\leq\left|\sum_{\ell\neq j}^{k}\frac{\braket{\tilde{v}^{(i)}}{\tilde{v}^{(\ell)}}\braketbra{\tilde{v}^{(\ell)}}{R^\dagger R}{\tilde{v}^{(j)}}}{\tilde{\sigma}_\ell^2}\right|
			+\alpha(1+\beta/\sigma_j^2)+\delta_{ij}\beta/\sigma_j^2
			+2\theta\kappa^2+ \alpha\\													
			&\leq (1+k\alpha)\beta\frac{5}{4}\kappa^2
			+2\theta\kappa^2+ 2\alpha\\			
			&\leq 2(\beta\kappa^2+\theta\kappa^2+ \alpha).								
		\end{align*}
		Let $e_\ell\in\C^k$ denote the $\ell$-th standard basis vector and let us define $V:=\sum_{\ell=1}^k\ketbra{\tilde{v}^{(\ell)}}{e_\ell}$.
		It follows that $\nrm{V^\dagger B V}\leq 2k(\beta\kappa^2+\theta\kappa^2+ \alpha)$. By \eqref{eq:apxVGram} we have that $\nrm{V^\dagger V-I}\leq k\alpha\leq1/4$, and thus $\nrm{(V^\dagger V)^{-1}}\leq 4/3$. By Lemma~\ref{lem:SpecBoundApx} we get that $\nrm{B}\leq 8k(\beta\kappa^2+\theta\kappa^2+ \alpha)/3$.
	\end{proof}

	If $\gamma\leq \frac{1}{10\kappa^2}$ and $\theta\leq \frac{1}{10\kappa^2}$, we get $\tilde{\sigma}_{\min}^2\geq \frac{4}{5\kappa^2}$.
	Then by Lemma~\ref{lem:SVecConv} we get that $\alpha\leq\frac{5}{4}\kappa^2\gamma$ and 
	$\beta\leq 3\kappa^2\gamma$. Substituting this into Equation~\eqref{eq:ultimateBound} we get the upper bound
	\begin{equation}\label{eq:ultimateBoundSubstituted}
		\gamma\left(8k\kappa^4 +10k\kappa^2/3\right) +
		\theta\frac{8k}{3}\kappa^2\leq \gamma 12k\kappa^4+
		\theta\frac{8k}{3}\kappa^2.
	\end{equation}
	Choosing $\theta= \frac{1}{16}\frac{\eps}{k\kappa^2}$ and $\gamma=\frac{1}{36}\frac{\eps}{k\kappa^4}$, the above bound~\eqref{eq:ultimateBoundSubstituted} becomes $\eps/2$.
	Therefore to succeed with probability at least $1-\eta/2$ it suffices to sample $r=2^{10}\ln\left(8n/\eta\right)\kappa^4 k^2\nrm{A}_F^2/\eps^2$ row indices, and then subsequently  $c=2^6\cdot 3^4\ln\left(8r/\eta\right)\kappa^8 k^2\nrm{A}_F^2/\eps^2$ column indices as shown by Theorem~\ref{thm:fastSample}. 

\subsection{The required precision for matrix element estimation}\label{subsec:precReq}
	Recall from Equation~\eqref{eq:goodEnoughSingAgain} that
	\vskip-5mm
	\begin{equation*}
	x'=\sum_{\ell=1}^k\frac{\lambda_\ell}{\tilde{\sigma}_\ell^2} \tilde{v}^{(\ell)},
	\end{equation*}	
	and $\tilde{x}$ is as above except we replace $\lambda_\ell$ with $\tilde{\lambda}_\ell$.
	As we argued in the beginning of the section, for the correctness of Algorithm~\ref{alg:LowRankHHL} it suffices to ensure $\nrm{\tilde{x}-x'}=\bigO{\eps}$, assuming that $\nrm{b}=1$. 	
	Now we show that if we have $\left|\lambda_\ell-\tilde{\lambda}_\ell\right|=\bigO{\frac{\eps\tilde{\sigma}_\ell^2\nrm{b}}{\sqrt{k}}}$, then the magnitude of perturbation can be bounded by $\bigO{\eps}$, and we also get that $\nrm{w}=\bigO{\kappa^2 \sqrt{k}}$. Let $e_\ell\in\C^k$ denote the $\ell$-th standard basis vector; we rewrite $\nrm{\tilde{x}-x'}$ as
	\begin{align*}
	\nrm{\sum_{\ell=1}^k\frac{\lambda_\ell-\tilde{\lambda}_\ell}{\tilde{\sigma}_\ell^2} \ket{\tilde{v}^{(\ell)}}}
	=\sqrt{\nrm{\sum_{\ell=1}^k\frac{\lambda_\ell-\tilde{\lambda}_\ell}{\tilde{\sigma}_\ell^2} \ket{\tilde{v}^{(\ell)}}\braket{e_\ell}{e_\ell}}^2}
	=\sqrt{\nrm{\left(\sum_{\ell=1}^k\ketbra{\tilde{v}^{(\ell)}}{e_\ell}\right)\left(\sum_{\ell=1}^k\frac{\lambda_\ell-\tilde{\lambda}_\ell}{\tilde{\sigma}_\ell^2}\ket{e_\ell}\right)}^2}.	
	\end{align*}
	Let us define $V:=\sum_{\ell=1}^k\ketbra{\tilde{v}^{(\ell)}}{e_\ell}$, and $\ket{z}:=\sum_{\ell=1}^k\frac{\lambda_\ell-\tilde{\lambda}_\ell}{\tilde{\sigma}_\ell^2}\ket{e_\ell}$, then we have that 
	\begin{align*}
	\nrm{\sum_{\ell=1}^k\frac{\lambda_\ell-\tilde{\lambda}_\ell}{\tilde{\sigma}_\ell^2} \ket{\tilde{v}^{(\ell)}}}
	=\sqrt{\bra{z}V^\dagger V \ket{z}}
	\leq \sqrt{\nrm{V^\dagger V}}\nrm{z}
	=\bigO{\eps},
	\end{align*}  	
	where we used that $\nrm{V^\dagger V}\leq 1 + k \alpha \leq \frac{4}{3}$ as we have shown in the proof of Lemma~\ref{lem:ApxProj}.
	
	Now we show that $\nrm{w}\!=\!\bigO{\kappa^2 \sqrt{k}}$. Remember that $\tilde{v}^{(\ell)}\!=\!\frac{R^\dagger w^{(\ell)}}{\tilde{\sigma}_\ell}$, thus $\tilde{x}\!=\!R^\dagger_{.\ell}\sum_{\ell=1}^k\frac{\tilde{\lambda}_\ell}{\tilde{\sigma}_\ell^3} w^{(\ell)}$. Let $w:=\sum_{\ell=1}^k\frac{\tilde{\lambda}_\ell}{\tilde{\sigma}_\ell^3} w^{(\ell)}\!$, then we get
	\begin{align*}
	\nrm{w}=\sqrt{\sum_{\ell=1}^k\frac{|\tilde{\lambda}_\ell|^2}{\tilde{\sigma}_\ell^2}}
	\leq \sqrt{\sum_{\ell=1}^k\frac{|\lambda_\ell|^2}{\tilde{\sigma}_\ell^6}}+\sqrt{\sum_{\ell=1}^k\frac{|\tilde{\lambda}_\ell-\lambda_\ell|^2}{\tilde{\sigma}_\ell^6}}
	\leq \bigO{\kappa^2}	\sqrt{\sum_{\ell=1}^k\frac{|\lambda_\ell|^2}{\tilde{\sigma}_\ell^2}}+\bigO{\kappa^2\eps}.		
	\end{align*}
	Finally observe that 
	\begin{align*}
	\sum_{\ell=1}^k\frac{|\lambda_\ell|^2}{\tilde{\sigma}_\ell^2}	
	&=\sum_{\ell=1}^k\frac{\braketbra{b}{A}{\tilde{v}^{(\ell)}}\braketbra{\tilde{v}^{(\ell)}}{A^\dagger}{b}}{\tilde{\sigma}_\ell^2}
	\leq \tr{\sum_{\ell=1}^k\frac{A\ket{\tilde{v}^{(\ell)}}\bra{\tilde{v}^{(\ell)}}A^\dagger}{\tilde{\sigma}_\ell^2}}
	=\tr{\sum_{\ell=1}^k\frac{\bra{\tilde{v}^{(\ell)}}A^\dagger A\ket{\tilde{v}^{(\ell)}}}{\tilde{\sigma}_\ell^2}}\\
	&\leq \sum_{\ell=1}^k\frac{\bra{\tilde{v}^{(\ell)}}R^\dagger R\ket{\tilde{v}^{(\ell)}}}{\tilde{\sigma}_\ell^2}+\bigO{k\kappa^2}\nrm{A^\dagger A-R^\dagger R}
	\leq \bigO{k + k \beta \kappa^2 + k \theta \kappa^2} \leq \bigO{k+\eps},
	\end{align*} 
	where the last two inequalities follow from Lemma~\ref{lem:ApxProj} and its follow-up discussion.
	\anote{It should be possible to bound this by $1+\bigO{\eps}$, removing the $\sqrt{k}$ factor -- but this is not very crucial and we can leave it for now.}

\section{Complexity of Algorithm~\ref{alg:LowRankHHL}}

	The complexity is dominated by two parts of the algorithm: finding the left singular vectors of an $r$ by $c$ matrix, and estimating some matrix elements of $A$.
	If we use naive matrix multiplication, then computing the singular value decomposition of $C C^\dagger$ costs
	\begin{equation*}
	\bigO{r^2c}=\bigOt{\kappa^{16} k^6\frac{\nrm{A}_F^6}{\eps^6}}.
	\end{equation*}
	
	In this section, we prove that this dominates the runtime of the algorithm.
	First, we use length-square sampling techniques similarly to Tang~\cite{tang2018QuantumInspiredRecommSys} to approximate the matrix elements $\lambda_\ell:=\braketbra{\tilde{v}^{(\ell)}}{A^\dagger}{b}$, which has complexity $\bigOt{\kappa^{8} k^4\frac{\nrm{A}_F^4}{\eps^4}}$ as we show in Section~\ref{subsec:IPEst}.
	Second, we show how to efficiently length-square sample from $\tilde{x}:=\sum_{\ell=1}^k\frac{ \tilde{\lambda}_\ell}{\tilde{\sigma}_\ell^2}\tilde{v}^{(\ell)}$ using rejection sampling.
	
\subsection{Length-square sampling techniques}
		
	\begin{definition}[Length-square distribution]
		For a non-zero vector $v\in \C^n$ we define the probability distribution $q^{(v)}$ on $[n]$ such that $q^{(v)}_i=\frac{|v_i|^2}{\nrm{v}^2}$.
	\end{definition}
	
	Note that if $v$ describes a normalized pure quantum state, the above distribution is exactly the distribution we get through measurement in the computational basis by the Born rule.
	
	\begin{definition}[Length-square access to a vector]
		We say that we have length-square access to the vector $v\in \C^n$ if we can request a sample from the distribution $q^{(v)}$ that takes cost $S(v)$.
		We also assume that we can query the elements of $v$ with cost $Q(v)$, and that we can query the value of $\nrm{v}$ with cost $N(v)$.\footnote{We assume for simplicity that $S(v), Q(v),$ and $N(v) \geq 1$.} We denote by $L(v):=S(v)+Q(v)+N(v)$ the overall access cost.
	\end{definition}	
	
	If the matrix $A$ is stored in an tree-like dynamic data structure~\cite{frieze2004FastMonteCarloLowRankApx,kerenidis2016QRecSys,tang2018QuantumInspiredRecommSys}, then the complexity of length-square accessing $A$ is $\bigO{\log(mn)}$. During the complexity analysis we will assume such efficient access.
	
	\begin{definition}[Length-square access to a matrix]
		We say that we have (row) length-square access to the matrix $A\in\C^{m\times n}$ if we have length-square access to the rows $A_{i.}$ of $A$ for all $i\in[m]$ and length-square access to the vector of row norms $a \in \R^m$, where $a_i:=\nrm{A_{i.}}$. 
		We denote by $L(A)$ the complexity of the length-square access to $A$.
	\end{definition}
	
	Note that length-square access to $A$ implies the ability to determine $\nrm{A}_F$ in $N(A)$ time.
	We will use that the closeness of two vectors in Euclidean distance implies closeness of their corresponding distributions.
	
	\begin{lemma}[Bounding Total Variation distance by Euclidean distance~{\cite[Lemma 6.1]{tang2018QuantumInspiredRecommSys}}]\label{lem:TVEuclid}
		For $v,w\in \C^n$, $\nrm{q^{(v)},q^{(w)}}_{TV}\leq \frac{2\nrm{v-w}}{\max(\nrm{v},\nrm{w})}$.
	\end{lemma}
	
\subsection{Estimating the matrix element \texorpdfstring{$\braketbra{\tilde{v}^{(\ell)}}{A^\dagger}{b}$}{<v|A|b>}}\label{subsec:IPEst}
	We use the inner product estimation method of Tang~\cite{tang2018QuantumInspiredRecommSys} for matrix element estimation.
	\begin{lemma}[Trace inner product estimation]\label{lem:traceIPEst}
		Suppose that we have length-square access to $A \in \C^{m\times n}$ and query access to the matrix $B\in\C^{m\times n}$ in complexity $Q(B)$.
		Then we can estimate $\tr{A^\dagger B}$ to precision $\xi\nrm{A}_F\nrm{B}_F$ with probability at least $1\!-\!\eta$ in time
		$$\bigO{\frac{\log(1/\eta)}{\xi^2}\left(L(A)+Q(B)\right)}.$$
	\end{lemma}
	\begin{proof}
		This exactly follows from \cite[Proposition 6.2]{tang2018QuantumInspiredRecommSys}, since $\tr{A^\dagger B}$ is the inner product of order-two tensors.
		Let $X$ be the random variable given by length-square sampling $i$ from $a$, the vector of row norms of $A$, sampling $j$ from $A_{i.}$, and setting the random output to $X = \frac{\nrm{A}_F^2}{A_{ij}}B_{ij}$.
		Then
		\begin{align*}
			\mathbb{E}[X] &= \sum_{i=1}^m\sum_{j=1}^n \frac{|A_{ij}|^2}{\nrm{A}_F^2}\frac{\nrm{A}_F^2}{A_{ij}}B_{ij}
			= \sum_{j=1}^n\sum_{i=1}^m (A^\dagger)_{ji}B_{ij}=\tr{A^\dagger B}\\
			\mathbb{E}[|X|^2] &= \sum_{i=1}^m\sum_{j=1}^n \frac{|A_{ij}|^2}{\nrm{A}_F^2}\frac{\nrm{A}_F^4}{|A_{ij}|^2}|B_{ij}|^2
			= \sum_{i=1}^m\sum_{j=1}^n \nrm{A}_F^2|B_{ij}|^2
			= \nrm{A}_F^2\nrm{B}_F^2.
		\end{align*}
		So $X$ is an unbiased estimator of $\tr{A^\dagger B}$.
		To compute the expectation, we use standard techniques: it suffices to estimate the the real and imaginary parts separately to additive precision $\xi\nrm{A}_F\nrm{B}_F/\sqrt{2}$ with success probability at least $1-\eta/2$.
		For each, we compute the mean of $\frac{9}{\xi^2}$ copies of $X$, and take the median of $6\log(2/\eta)$ such empirical mean estimators, giving the desired result.
	\end{proof}
	We can estimate $\lambda_\ell=\braketbra{\tilde{v}^{(\ell)}}{A^\dagger}{b}=\tr{\braketbra{\tilde{v}^{(\ell)}}{A^\dagger}{b}}=\tr{A^\dagger\ketbra{b}{\tilde{v}^{(\ell)}}}$ using this lemma. Observe that $\nrm{\ketbra{b}{\tilde{v}^{(\ell)}}}_F=\nrm{\tilde{v}^{(\ell)}}\nrm{b}\leq (1+\eps)\nrm{b}$, and we can query the $(i,j)$ matrix element of $\ketbra{b}{\tilde{v}^{(\ell)}}$ by querying $b_i$ and $\tilde{v}^{(\ell)}_j$, which has $\bigOt{1}$ and $\bigOt{r}$ cost respectively.
	We desire to estimate $\lambda_\ell$ to additive precision $\bigO{\frac{\eps\sigma_\ell^2\nrm{b}}{\sqrt{k}}}$ with success probability $\frac{\eta}{2k}$.
	By applying Lemma~\ref{lem:traceIPEst}, we can compute such an estimate $\tilde{\lambda}_\ell$ with complexity
	\begin{equation*}
	\bigOt{\log(2k/\eta)\frac{k\nrm{A}_F^2}{\eps^2\sigma_{\ell}^4}r}=\bigOt{\kappa^4 k\frac{\nrm{A}_F^2}{\eps^2}r}=\bigOt{\kappa^{8} k^3\frac{\nrm{A}_F^4}{\eps^4}}.
	\end{equation*}

\subsection{Sampling from the approximate solution}

	Our goal is to sample from the length-square distribution of $\tilde{x}=R^\dagger w$. 
	In order to tackle this problem we invoke a result from~\cite{tang2018QuantumInspiredRecommSys} about length-square sampling a vector that is a linear-combination of length-square accessible vectors. For completeness we present its proof too, following the approach of Tang~\cite{tang2018QuantumInspiredRecommSys}.
	\begin{lemma}[Length-square sample a linear combination of vectors~{\cite[Proposition 6.4]{tang2018QuantumInspiredRecommSys}}]\label{lem:linCombSamp}
		Suppose that we have length-square access to $R\in \C^{r\times n}$ having normalized rows, and we are given $w\in\C^r$ (as a list of numbers in memory). Then we can implement queries to the  vector $y:=R^\dagger w\in \C^n$ with complexity $Q(y)=\bigO{rQ(R)}$ and we can length-square sample from $q^{(y)}$ with complexity $S(y)$ such that $\mathbb{E}[S(y)]=\bigO{\frac{r\nrm{w}^2}{\nrm{y}^2}(S(R)+r Q(R))}$.
		\anote{Maybe describe the cost of computing the length of the vector as well.}
	\end{lemma}
	\begin{proof}
		The algorithm is simple, it proceeds by rejection sampling: one should first sample a row index $i\in [r]$ uniformly, then draw a column index $j$ distributed as $|R_{ij}|^2$, then compute $|y_j|^2=|\braket{w}{R_{.j}}|^2$, $\nrm{R_{.j}}^2$ and either output $j$ with probability $\frac{|\braket{w}{R_{.j}}|^2}{\nrm{w}^2\nrm{R_{.j}}^2}$ or sample $(i,j)$ again. In each round the probability that we pick the column index $j$ is $\sum_{j=1}^r|R_{ij}|^2/r=\nrm{R_{.j}}^2/r$, and the probability that we output $j$ is $|y_j|^2/(r\nrm{w}^2)$. The success probability in each round is $\nrm{y}^2/(r\nrm{w}^2)$, therefore the expected number of rounds is $r\nrm{w}^2/\nrm{y}^2$.
	\end{proof}
 	Since all rows of $R$ have norm $\nrm{A}_F/\sqrt{r}$, and $\nrm{\tilde{x}}=\Omega(1)$ by Lemma~\ref{lem:linCombSamp} we can length-square sample from $\tilde{x}$ in expected complexity 
	\begin{equation*}
		 \bigO{\frac{r\nrm{w}^2\nrm{A}_F^2/r}{\nrm{\tilde{x}}^2}r}
		= \bigO{\frac{\nrm{w}^2\nrm{A}_F^2}{\nrm{\tilde{x}}^2}r}
		=\bigO{\kappa^{4} k\nrm{A}_F^2r}=\bigOt{\frac{\kappa^{8}k^3\nrm{A}_F^4}{\eps^2}}.
	\end{equation*}
	
\section{Discussion}
	
	We presented a proof-of-principle algorithm for approximately inverting low-rank matrices in runtime that is logarithmic in the dimensions. For simplicity we analysed the case when the matrix has low rank, however it should be possible to devise similar results when the matrix does not have low rank, but one only intends to invert the matrix on a ``well-conditioned'' subspace. We expect that the complexity can be further improved by using optimized algorithms for finding an approximate singular vale decomposition of the subsampled matrix $C$. Also, one might use another variant of the algorithm where the left singular vectors of $A$ are approximated using some variant of the FKV algorithm~\cite{frieze2004FastMonteCarloLowRankApx}, instead of the right singular vectors. Another approach could be to use a different low-rank approximation method, which reconstructs an approximation of $A$ by using a linear combination of rows and columns such as described in~\cite{kannan2017RandAlgNumLinAlg}.
	
	Although in this paper we focus on implementing the pseudo-inverse of a matrix by inverting the singular values, one could in principle apply any desired function to the singular values. It has recently been shown that many quantum algorithms can be expressed as a singular value transformation problem~\cite{gilyen2018QSingValTransf}. This supports Tang's suggestion~\cite{tang2018QInspiredClassAlgPCA} that many quantum algorithms can be effectively turned to randomized classical algorithms via length-square sampling techniques incurring only polynomial overheads. Our work gives evidence that this conversion can be done in general for low-rank problems, suggesting that exponential quantum speed-ups are tightly related to problems where high-rank matrices play a crucial role, like in Hamiltonian simulation or the Fourier transform. However, more work remains to be done on understanding the class of problems for which exponential quantum speed-up can be achieved.

\vspace{-2mm}
	 
\section*{Acknowledgments}
	A.G. thanks Márió Szegedy for introduction to the problem and sharing insights, and Ronald de Wolf for helpful comments on the manuscript. 	S.L was supported by ARO and OSD under a Blue Sky Program.

\vspace{-1mm}
	
\bibliographystyle{alphaUrlePrint}
\bibliography{Bibliography}

\vspace{-3mm}

\appendix

	\providecommand\mywordcount{
		\ifcount
		$ \phantom{\sum}$ \\ \noindent\textbf{\large Wordcount} \\ $\phantom{\sum}$ \\
		\noindent\input|"if [ -f /home/gilyen/texcount.pl ]; then /home/gilyen/texcount.pl -sub=section \jobname.tex | grep -e Words -e Number -e Section -e top -e Part | awk 1 ORS='\string\\\string\\' | sed -e 's/\string\_/ /g'; else texcount -sub=section \jobname.tex | grep -e Words -e Number -e Section -e top -e Part | awk 1 ORS='\string\\\string\\' | sed -e 's/\string\_/ /g'; fi"
		text+headers+captions (\#headers/\#floats/\#inlines/\#displayed)\\
		\else
		\fi
	}
	
	
\end{document}